\documentclass[10pt,pra,superscriptaddress,twocolumn]{revtex4-2}
\usepackage[utf8]{inputenc}
\usepackage{CJKutf8}
\usepackage{multirow}
\usepackage{subcaption}
\usepackage{amsmath}
\usepackage{times}
\usepackage{booktabs}
\usepackage[dvipsnames]{xcolor}

\usepackage{tikz}
\usepackage{amssymb}
\usepackage{amsthm}
\usepackage{verbatim}

\usepackage{braket}
\usepackage{graphicx}
\usepackage{caption}
\usepackage{subcaption}
\usepackage{bbold}
\usepackage[hyperindex,breaklinks,colorlinks=true,citecolor=blue,urlcolor=blue]{hyperref}
\usepackage{hhline}

\newcommand{\cH}{\mathcal{H}}

\newcommand{\cK}{\mathcal{K}}
\newcommand{\cL}{\mathcal{L}}

\definecolor{nred}{rgb}{0.9,0.1,0.1}
\definecolor{nblack}{rgb}{0,0,0}
\definecolor{nblue}{rgb}{0.2,0.2,0.8}
\definecolor{ngreen}{rgb}{0.2,0.6,0.2}

    \newcommand{\tibc}{\text{IBC}}

\newtheorem{cor}{Corollary}
\newtheorem{definition}{Definition}

\newtheorem{thm}{Theorem}
\newtheorem{Lemma}{Lemma}

\newcommand{\tr}{{\rm tr}}
\usepackage[capitalize]{cleveref}
\begin{document}
\title{Connection between Contextuality-Breaking and Incompatibility-Breaking Qubit Channels}
\author{Swati Kumari}
\affiliation{Department of Physics, Indian Institute of Technology Dharwad, Dharwad, Karnataka 580011, India}
\affiliation{Department of Physics, Patna Science College, Patna University, Bihar 800005, India}
\author{Sumit Mukherjee}
\email{mukherjeesumit93@gmail.com}
\affiliation{Department of Physics, Indian Institute of Technology Dharwad, Dharwad, Karnataka 580011, India}
\affiliation{Department of Physics of Complex Systems, S. N. Bose National Center for Basic Sciences, Block JD, Sector III, Salt Lake, Kolkata 700106, India}
\author{R. Prabhu}
\affiliation{Department of Physics, Indian Institute of Technology Dharwad, Dharwad, Karnataka 580011, India}
\date{today}
\begin{abstract}
Contextuality and measurement incompatibility are two fundamental aspects of nonclassicality, and their manifestations in observed quantum correlations are often deeply interconnected. Recently, measurement incompatibility has been studied in connection with nonlocality, particularly in terms of their robustness under various quantum channels. This line of investigation helps to establish a connection between the channels that break nonlocality and those that break incompatibility. In this study, we focus on an asymmetric bipartite Bell scenario involving three and four inputs on Alice’s and Bob’s sides, respectively, with each of these inputs having dichotomous outcomes. Under the assumption of locality, the observed statistics in this asymmetric scenario obeys the Elegant-Bell inequality (EBI). Here, we use a different version of the EBI that relies on the assumption of the preparation noncontextuality. By taking the violation of this noncontextual version of EBI as a witness of preparation contextuality we establish a connection between the channels that break contextuality and the channels that break triple-wise measurement incompatibility. Our results suggest that any channel which breaks EBI contextuality will also break Clauser-Horne-Shimony-Holt (CHSH) contextuality; however, the reverse does not hold. We also show that a depolarising channel that breaks N-wise incompatibility can also break a certain form of contextuality, witnessed by a generalised inequality involving N measurements on one wing of a bipartite Bell scenario.
\end{abstract}
\maketitle
\section{Introduction} 
Two of the most paradigmatic manifestations of quantum correlations, which stand in contrast to our classical worldview, are nonlocality \cite{bell1964einstein,Brunner2014bell} and contextuality \cite{contextualityReview,KochenS}.
 While the former rules out any local realistic description of quantum correlations, the later demonstrates their incompatibility with deterministic value assignments to the outcomes that are independent of the measurement context. The idea of contextuality was latter generalised to preparations and transformations of quantum systems as well \cite{spekkcontext} . Such generalization of contextuality is essentially conceptualized based on the incompatibility of quantum statistics with ontological models that admit same ontic description for equivalent operational procedures \cite{spekkcontext}. This form of contextuality is widely used in various quantum information processing protocols \cite{disSchmid,Spekkens2009,Pan2019,Saha2019,mukherjee21}, quantum computation \cite{stabilizer} and metrology \cite{lostaglio}.

 Another fundamental feature in quantum theory is measurement incompatibility—a property that prevents all observables from possessing simultaneous definite values. For the case of ideal sharp measurements \cite{nielsenchuangbook}, if the operators corresponding to two observables commute, they are said to be co-measurable or compatible. In such cases, a measurement is represented by a set of projectors. However, the most general form of a measurement is expressed through the positive operator valued measures (POVMs) \cite{Busch1986unsharp}. A set of POVMs is said to be jointly measurable (or compatible) if there exists a single global POVM whose marginals reproduce each element of the set. If no such global POVM exists, the set is incompatible \cite{incomRevModPhys.95.011003}. Measurement incompatibility is a broader concept than commutativity: while commutativity implies compatibility, the converse is not true. For projective measurements, pairwise commutativity of an arbitrary set of $N$ observables guarantees their global commutativity. In contrast, for POVMs, pairwise compatibility of $N$ measurements associated with non-commuting observables does not, in general, imply global compatibility \cite{KunjwalJM,kunjwalMS}.

Quantum channels \cite{Heinosaari_book_QF} play a pivotal role in understanding the robustness of quantum correlations like nonlocality, contextuality, as well as properties of measurements such as measurement incompatibility. A quantum channel is a completely positive trace-preserving (CPTP) map that models the interaction of a quantum system with its environment, often leading to decoherence or loss of nonclassical resources~\cite{Heinosaari_book_QF}. In this context, a \textit{nonlocality-breaking channel} (NBC) \cite{pal2015non} is one whose action on a subsystem of any bipartite system ensures that the resulting correlations admit a local hidden variable model, thereby satisfying the Clauser-Horne-Shimony-Holt (CHSH) inequality \cite{Clauser_1969}. Analogously, an \textit{incompatibility-breaking channel} (IBC) maps any given set of incompatible observables to a jointly measurable set~\cite{heinosaari2015}. Extending this idea, a \textit{contextuality-breaking channel} (CBC) is defined as one that, when acting on a subsystem of a bipartite scenario renders the statistics compatible with a preparation noncontextual model~\cite{Mukherjee_2024,GoraiAKP_2018}. Since contextuality in such settings has a one-to-one correspondence with measurement incompatibility~\cite{uola20}, studying CBCs provides an opportunity to establish a direct link between the dynamical degradation of contextuality and incompatibility under the action of quantum channels. This perspective also allow us to investigate whether channels that destroy nonlocal correlations also necessarily destroy the underlying contextuality and vice-versa.
 
 Recent works have revealed deep operational connections between different manifestations of quantum nonclassicality. For example, the measurement incompatibility is necessary and sufficient for the violation of CHSH inequality \cite{wolf2009measurements} in the bipartite scenario involving two random measurements per party. Therefore, the set of qubit unital channels that are nonlocality breaking coincides exactly with the set of incompatibility breaking channels \cite{Swati2023}. However, this equivalence breaks down when the number of input measurement settings or outputs are increased \cite{bene2018measurement,hirsch18}. In such extended configuration, nonlocality and incompatibility do not maintain such one-to-one connections, indicating a richer structure in the resource theory of quantum measurements. On the other hand contextuality, in particular, the preparation contextuality has been shown to admit a one-to-one correspondence with measurement incompatibility even in the extended operational scenarios \cite{uola20}. 
 
 We investigate the connection between CBC and IBC by extended operational scenario, where Alice and Bob perform three and four dichotomic measurements respectively. The relevant nonclassicality witness in this setting is the Elegant-Bell inequality (EBI) \cite{gisin2007bell}. Unlike the CHSH inequality, the noncontextual and local bounds of Bell functional corresponding to the EBI are different. Here we employ a version of the EBI whose upper bound is determined by the assumption of preparation noncontextuality. A violation of this bound serves as a sufficient, though not necessary, condition for contextuality. By examining how different classes of noisy channels influence EBI violations, we establish an operational connection between CBC and IBC~\cite{heinosaari2015}, extending the analysis to scenarios beyond CHSH.

We introduce the formal definition of an \emph{EBI-contextuality-breaking channel} (EBI-CBC) and compare its behavior with CHSH-contextuality-breaking channels (CHSH-CBC) across several important noise models, including depolarizing, amplitude damping, loss, and dephasing channels. Our results show that every EBI-CBC is CHSH-CBC, but the converse is not true, thereby establishing a strict inclusion relation between these two classes of channels. For unital qubit channels, we find a one-to-one correspondence between EBI-CBC and triple-wise IBC. In addition, the analysis of different channel actions on contextual correlations are supplemented with an investigation of the white-noise robustness of the resulting correlations. Finally we generalise our results by establishing the connection between $N$-wise incompatibility breaking channel ($N$-IBC) and CBC in bipartite Bell scenario involving $N$ measurements at one side. Our quantitative comparison of channel parameters of depolarizing channel across varying numbers of extended measurement settings shows that increased number of measurement settings increases the noise threshold required to destroy nonclassicality, refining our understanding of how quantum resources behave under such extensions.

This paper is organized as follows. Sec.~\ref{subsec:contx} presents the definitions for nonlocality, contextuality, NBC, CBC, and EBI. In Sec.~\ref{subsec:cbreaking}, we derive the conditions under which EBI contextuality is broken and establish its connection with certain incompatibility breaking channels. In Sec.~\ref{sec:white_robustness} we discuss the white noise robustness of contextuality. We also discuss the quantitative characteristics of different channels in terms of their ability to break EBI contextuality and CHSH contextuality. In Sec.~\ref{subsec:gen}, we generalise our results to the case of an arbitrary number $N$ of measurements. Finally, in Sec.~\ref{sec:discussion} we discuss the implications of our findings and outline the possible future directions.
 
 \section{Preliminaries}
 \label{subsec:contx}
  \subsection{Contextuality in Bipartite Bell-type Scenario}

We briefly review the ontological model \cite{contextualit05} of an operational theory, along with the assumption of \textit{preparation noncontextuality}. An operational theory consists of a set of preparations $\{P\}$ and measurements $\{M\}$, which are collectively designed to predict the experimental statistics via $p(k|P, M)$, i.e., the probability of obtaining outcome $k$ given a specific preparation and measurement. As an operational theory, quantum mechanics represents a preparation procedure by a density matrix $\rho$, while measurements are generally described by the POVMs $\{E_k\}$, satisfying $\sum_k E_k = \mathbb{I}$. The corresponding quantum probability of obtaining a particular outcome $k$ is given by the Born rule: $p(k|P, M) = \mathrm{Tr}[\rho E_k]$.

While quantum theory successfully explains most of the observed atomic and subatomic phenomena, it does not provide an objective description of the properties of a physical system. The ontological model framework \cite{Harrigan2010-cb} aims to reconcile quantum theory with such objective descriptions. However, as proved in different no-go theorems \cite{Brunner2014bell,contextualityReview}, to construct such a model we must defenestrate our beliefs about local reality or context-independent reality, leading to the emergence of nonlocal \cite{bell1964einstein} and contextual correlations. In ontological model, a preparation procedure $P$ is associated with a probability distribution $\mu_{P}(\lambda)$ over the ontic state space $\Lambda$, satisfying the normalization condition $\int_{\Lambda} \mu_{P}(\lambda) \, d\lambda = 1$ for all $P$, with $\lambda \in \Lambda$. Here, $\lambda$ represents the ontic (i.e., underlying) state of the system. Given a POVM element $E_k$ of measurement $M$, each ontic state $\lambda$ is associated with a response function $\xi(k | \lambda, M)$ satisfying $\sum_k \xi(k | \lambda, M) = 1$ for all $\lambda$.

Any ontological model consistent with quantum theory must reproduce the quantum statistics obtained via the Born rule:
\begin{equation}
\int_{\Lambda} \mu_{P}(\lambda | \rho) \, \xi(k | \lambda, M) \, d\lambda = \mathrm{Tr}(\rho E_k),
\label{eq:born}
\end{equation}
where $\mu_{P}(\lambda | \rho)$ is the probability of an ontic state $\lambda$ given the quantum state $\rho$ is prepared by preparation procedure $P$. To introduce the notion of preparation noncontextuality, we consider operationally equivalent preparation procedures---those that cannot be distinguished by any measurement. Specifically, two preparation procedures $P_0$ and $P_1$ are operationally equivalent if $p(k | P_0, M) = p(k | P_1, M)$ for all $M$ and $k$. In quantum mechanics, this corresponds to two distinct procedures that prepare the same density matrix $\rho$, which cannot be distinguished by any measurement.

An ontological model is said to be \textit{preparation noncontextual} if operationally equivalent preparations are represented by identical ontic state distributions ~\cite{contextualit05}, i.e., $\forall M,k,$
\begin{equation}
\label{eq:noncontextuality}
 \ p(k | P_0, M) = p(k | P_1, M) \Rightarrow \mu_{P_0}(\lambda | \rho) = \mu_{P_1}(\lambda | \rho).
\end{equation}

We refer to the violation of this assumption by quantum theory as preparation contextuality. Here we use the term \textit{(non) contextuality} to refer this form of (classical) quantum correlations.

To understand contextuality in a bipartite scenario, we consider two observers, namely Alice and Bob, who share a nontrivially correlated system and perform measurements $A_{x}$ and $B_{y}$ on their respective subsystems, with $x\in [m]$ and $y\in [n]$. Here  $[l]$ represents a set of numbers $\{1,2....l\}$. This scenario can be thought of as an entanglement assisted prepare-measure scenario, where Alice's measurement remotely prepares the local state on Bob's side via quantum steering \cite{uola20}, and Bob performs measurements on that state. However, due to no-signaling the state produced at Bob's wing for different measurements at Alice's wing will be operationally equivalent. Now, these operationally equivalent preparations can be represented by same ontic states in the ontological model as described in Eq. \eqref{eq:noncontextuality}. This essentially is the assumption of noncontextuality. Now, for $x,y\in \{1,2\}$, the joint outcome statistics of Alice and Bob in a noncontextual model obeys Bell-CHSH inequality as given bellow \cite{Clauser_1969}.
  \begin{equation}
 \label{Bell-op}
 \mid \langle \mathcal{B} \rangle \mid = \mid \langle{A}_1\otimes({B}_1 + {B}_2) \rangle + \langle{A}_2\otimes({B}_1-{B}_2) \rangle\mid \leq 2
 \end{equation}

 Two-qubit entangled state and suitable choice of incompatible observables at both wings lead to the violation of the inequality given by Eq.\eqref{Bell-op}. Here, entanglement is necessary but not sufficient for the violation of the Bell inequality; however, measurement incompatibility is both necessary and sufficient for the violation of the Bell-CHSH inequality~\cite{wolf2009measurements}.

Note that, the local bound on the operator expectation value $\mid \langle \mathcal{B} \rangle \mid $ is same as the noncontextual bound stated in Eq. \eqref{Bell-op}. In other word, for the symmetric $x,y\in \{1,2\}$ case, the noncontextual and local bound coincide. Therefore, CHSH-nonlocality and CHSH-contextuality are equivalent. However, the situation is different for a bipartite scenario involving measurements $m \neq n,$ with $ m,n \geq 2$. In this paper we concentrate in such a bipartite scenario with $m=3$ and $n=4$.

\textit{EBI and its maximum quantum violation:}
Consider Alice and Bob perform their individual dichotomic measurements on two spatially separated nontrivially correlated qubits of observables, $A_{x}$ and $B_{y}$, with $x\in \{1,2,3\}$ and $y\in \{1,2,3,4\}$. We call this a (3422)-Bell scenario. Now, if the joint statistics of Alice and Bob obey the noncontextuality assumption then it satisfies the following inequality ~\cite{Mukherjee_2024, GoraiAKP_2018},
 \begin{align}
 \label{eq:EBI}
 \nonumber
 \langle \mathfrak{B}\rangle  &=  \langle A_{1}\otimes(B_{1}+B_{2}-B_{3}-B_{4}) \rangle \\
 \nonumber       &+ \langle A_{2}\otimes(B_{1}-B_{2}+B_{3}-B_{4}) \rangle \\
                 &+ \langle A_{3}\otimes(B_{1}-B_{2}-B_{3}+B_{4}) \rangle  \overset{\text{NC}}{\leq} 4.
 \end{align}
 
 Note that, the operator expectation value of the RHS of Eq. \eqref{eq:EBI} is bounded by $\langle \mathfrak{B}_{L} \rangle \leq 6 $ for statistics admitting locality \cite{gisinelegent}. This particular form of inequality is famously called the EBI. However, in this paper we use this noncontextual version of the EBI given by Eq. \eqref{eq:EBI}, i.e., $\langle \mathfrak{B}_{NC} \rangle \leq 4 $ to establish our main results.
 
 The optimal quantum value $\mathfrak{B}_{Q}^{opt}$ of Bell functional $\mathfrak{B}_{Q}$ is obtained for the shared maximally entangled state, observables represented by mutually orthogonal vectors on the Bloch sphere at Alice's wing, and four vectors on the Poincare surface ~\cite{gisin2007bell} at Bob's wing. For example, a specific choice of observables at Alice's wing, 
 \begin{equation}\label{eq:alice_obs}
      A_1=\sigma_{z}, \ \  A_2=\sigma_{y}, \ \   A_3=\sigma_{x},
 \end{equation}
 and at Bob's wing, 
 \begin{eqnarray}\label{eq:bob_obs}
    && B_1=(\sigma_{x}-\sigma_{y}+\sigma_{z})/\sqrt{3}; \  B_2=(\sigma_{x}+\sigma_{y}-\sigma_{z})/\sqrt{3}  \\
     && B_3=(-\sigma_{x}-\sigma_{y}-\sigma_{z})/\sqrt{3}; \  B_4=(-\sigma_{x}+\sigma_{y}+\sigma_{z})/\sqrt{3},\nonumber
 \end{eqnarray}
 
  and for two-qubit maximally entangled state the maximum quantum violation of the inequality in Eq. \eqref{eq:EBI} is achieved as $\mathfrak{B}_{Q}^{opt}=4\sqrt{3}$.
  
\subsection{Measurement incompatibility}
 In quantum theory, it is generally not possible to perform all measurements jointly, that is, there exist measurements whose outcomes cannot be determined simultaneously within a single experimental setup. For projective measurements, joint measurability is equivalent to commutativity; that is, a set of projective measurements is jointly measurable if and only if the corresponding operators commute. However, commutativity is only a sufficient condition and not a necessary one for joint measurability. For example, non-commuting observables can still admit a joint probability distribution when one considers general measurements given by the POVMs.
 
A quantum measurement $M \in \mathcal{M} (\mathcal{H})$ with outcome set $\Omega_{M}$ is defined as a map: $m\rightarrow M^{m}$ whose action is to assign a positive operator for every outcome $a \in \Omega_{M}$.The operator $M^{m} \in \mathcal{L} (\mathcal{H})$ is the relevant POVM element corresponding to the outcome $a$ of measurement $M$ satisfying  $M^{m}\geq 0 \ \forall m$, and $\sum_{m} M^{m}=\mathbb{I}$, with the set $\{M^{m}\}_{m}$ being the POVM. Here, $\mathcal{L}(\mathcal{H})$ and $\mathcal{M}(\mathcal{H})$ represent the set of all linear operators  and the set of all measurements acting on the Hilbert space $\mathcal{H}$. Given the density matrix $\rho$ corresponding to a quantum state, the probability of a particular outcome $m$ of measurement $M$ is given by $p(m)= tr[M^{m}\rho]$.  Now, a set $X=\{M_{1}, M_{2},\ldots, M_{n}\}$ with $X \subset \mathcal{M} (\mathcal{H})$ is said to be jointly measurable or compatible if there exists a global POVM $G^{m_{1},m_{2},\ldots,m_{n}}$ defined on the product outcome space $ \Omega = \Omega_{1} \times \Omega_{2} \times \cdots \times \Omega_{n}$, where each outcome $m_i\in\Omega_i$, $i=1,\ldots,n$,  satisfying the following properties:
\begin{enumerate}
    \item  $G^{m_{1},m_{2},\cdots,m_{n}}\geq 0 \hspace*{0.3cm}\forall \  m_{1},m_{2},\cdots,m_{n}$ (Positivity)
    \item $\sum_{m_{1},m_{2},\cdots,m_{n}} G^{m_{1},m_{2},\cdots,m_{n}}=1$ (Completeness)
    \item  $M_{k}^{m_{k}}=\sum_{m_{1},\cdots, m_{k-1}, m_{k+1},\cdots,m_{n}}G^
    {m_{1},m_{2},\cdots,m_{n}}$ (Reproducibility of all the individual measurements).
\end{enumerate}
   
 Otherwise, the set of measurements is called incompatible.

 \subsection{Quantum Channels}
In general, quantum systems are inevitably subject to interactions with various types of noisy environments. These interactions can introduce disturbances that alter the system's behavior, ultimately diminishing its reliability and effectiveness in performing different quantum information processing tasks. Such interactions of the state with the environment are best described in terms of quantum channels \cite{Heinosaari_book_QF}.
 \vskip 0.1cm
 In Schrodinger picture, a quantum channel (QC) $\mathcal{E}:\cL(\mathcal{H})\rightarrow \cL(\mathcal{K})$, is a completely positive trace-preserving (CPTP) linear map that acts on a quantum state in the Hilbert space $\mathcal{H}$ and transforms it into another quantum state belonging to Hilbert space $\mathcal{K}$ \cite{Heinosaari_book_QF}. It is important to note that, in general, the dimension of the input and output Hilbert spaces may not be the same. Conventionally, the action of a quantum channel is well expressed in terms of the Krauss representation~\cite{nielsenchuangbook}, according to which the output state after the action of a quantum channel $\mathcal{E}$ is given by, $\rho^\prime =\mathcal{E}\big(\rho\big) =\sum_{j=1}^{n}K_j\rho K_j^\dagger$, where Krauss operators $K_j$ satisfy $\sum_j K^\dag_j K_j=\mathbb{I}$.

  In Heisenberg picture the action of $\mathcal{E}$ is denoted by $\mathcal{E}^{*}:\cL(\cK)\rightarrow\cL(\cH)$ and is defined as,

\begin{align}
\label{eq:SHpictures}
    \tr[\mathcal{E}(\rho)M^{m}]=\tr[\rho\mathcal{E}^{*}(M^{m})] \ \ \forall \rho\in\cL(\cH),M^{m}\in\cL(\cK).
\end{align}
The key insight from Eq. (\ref{eq:SHpictures}) is that the noise leading to decoherence in a state-based resource can equivalently be interpreted as a distortion of the measurement-based resource. This idea enables the formulation of fundamental connections between the properties of the states and the measurements in a more general setting. In particular, consider a bipartite scenario discussed earlier involving parties Alice and Bob sharing an entangled state $\rho_{AB}\in\cL(\cH_{A} \otimes \cH_{B})$ and performing measurements $A_{x}$ and $B_{y}$ with $x,y \in \{0,1\}$ on their respective subsystems. Then in such a scenario Eq. \eqref{eq:SHpictures} can be rewritten as,
\begin{align}
\label{eq:SHpictures1}
 &\forall \rho_{AB}\in\cL(\cH_{A} \otimes \cH_{B}),A_{x}^{a_x}\otimes B_{y}^{b_y}\in\cL(\cK_{A} \otimes \cK_{B}),\nonumber \\
   & \tr[(\mathcal{E} \otimes \mathbf{I}) (\rho_{AB})A_{x}^{a_x}\otimes B_{y}^{b_y}]=\tr[\rho_{AB}(\mathcal{E}^{*}\otimes \mathbf{I})(A_{x}^{a_x}\otimes B_{y}^{b_y})], 
\end{align}

with $\mathbf{I}$ being the identity channel. Motivated by this, the relationship between the channels that break certain bipartite quantum correlations, viz., steering and nonlocality, by introducing decoherence to the state and those that break measurement incompatibility have been investigated in  \cite{heinosaari2015,Swati2023,pal2015non}.

 \begin{figure}
 \begin{align}
\centering
\includegraphics[width=85mm]{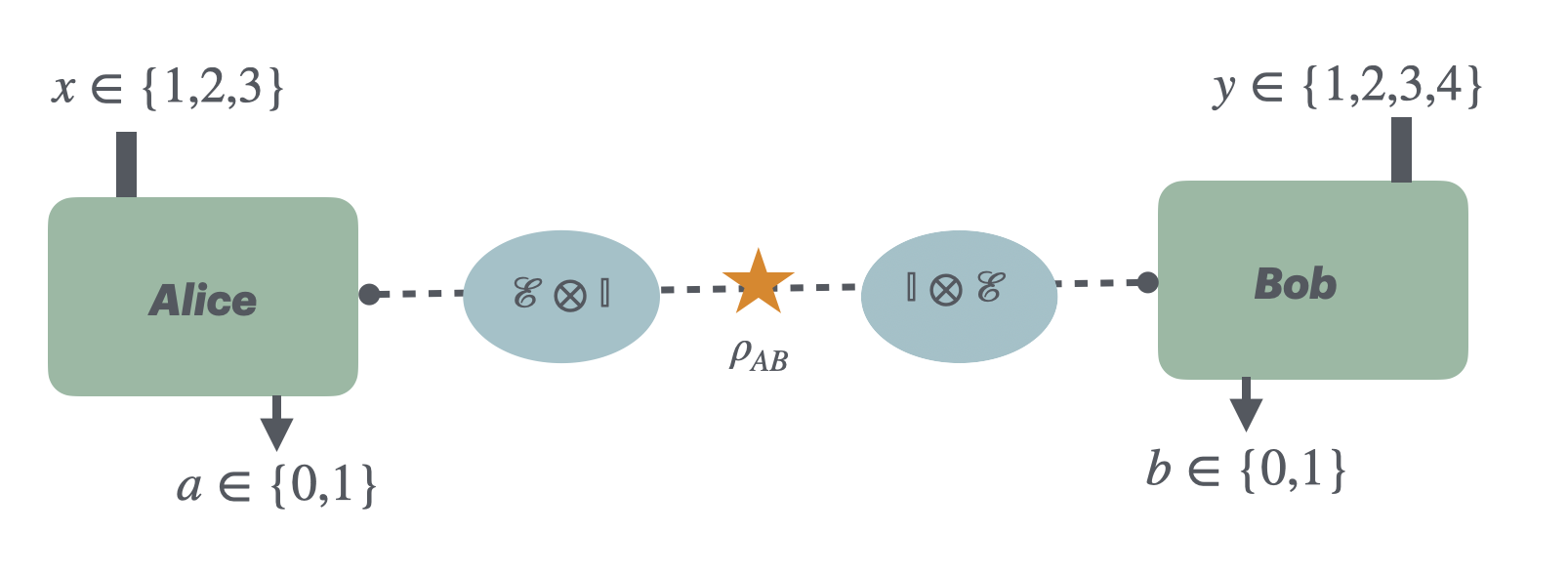}\nonumber
\end{align}
\caption{Figure depicts the action of contextuality breaking channels ($\mathcal{E}\otimes\mathbf{I}$) and ($\mathbf{I} \otimes \mathcal{E} $) on first and second qubit of the shared state $\rho_{AB}$ between Alice and Bob respectively.}
\label{NBCpicture}
\end{figure}
\vskip 0.3cm
 
\textit{Nonlocality and contextuality breaking channels:} Consider an arbitrary bipartite quantum state $\rho_{AB}$ that violates the Bell-CHSH inequality. 
If after the action of a channel on one subsystem, the resulting state 
$\rho'_{AB} = (\mathbf{I} \otimes \mathcal{E})(\rho_{AB})$ satisfies the Bell-CHSH inequality, then the channel $\mathcal{E}$ is said to be a CHSH nonlocality-breaking channel (CHSH-NBC). For a state $\rho'_{AB}$ satisfying the Bell-CHSH inequality, if the POVMs used for the measurement of the subsystems of the shared bipartite system of Alice and Bob are denoted as $\{E_{x}^{a_x}\}$ and $\{E_{y}^{b_x}\}$ respectively, we can always define conditional distributions $P(a|x,\lambda)$ and $P(b|y,\lambda)$ as well as a shared classical variable $\lambda$, such that the joint statistics can be written as,
 \begin{eqnarray}
 \label{eq:factor}
 p(a_{x},b_{y}|A_{x},B_{y})=& \hspace{-0.8cm}\tr\left[( E_{x}^{a_x}\otimes E_{y}^{b_y} )\rho'_{AB} \right] \nonumber \\
 =&\int d\lambda ~p(\lambda) ~P(a|x,\lambda) ~P(b|y,\lambda).
 \end{eqnarray}
 On the other hand, if the measurement statistics cannot be expressed in the factorized form of Eq.~\eqref{eq:factor}, 
then the Bell-CHSH inequalities are violated, implying that $\rho'_{AB}$ is nonlocal. The action of a nonlocality-breaking channel is illustrated in Fig.~\ref{NBCpicture}, 
where a quantum channel $\mathcal{E}$ acts on the two subsystems in a bipartite Bell scenario. It is intuitively evident that if a channel breaks the nonlocality of a maximally entangled state, then it will also break the nonlocality of any bipartite entangled state~\cite{pal2015non}.
It is important here to note that in a bipartite scenario, having two dichotomic measurements on both sides, a nonlocality-breaking channel is always a contextuality-breaking channel, i.e., CHSH-NBC is equivalent to CHSH-CBC. This is because, as discussed earlier, the criteria of witnessing nonlocality and contextuality are the same (as given by Eq. \eqref{Bell-op}) for such a scenario. Then the question arises does this equivalence remains intact in general scenarios beyond two inputs? To address this, let's first discuss what we mean by a contextuality-breaking channel.

We call a channel contextuality-breaking if, when applied to the input state of a bipartite scenario, the resulting joint statistics satisfy a noncontextuality-witnessing inequality, whereas the joint statistics generated by the original state violate the same inequality. It is noted in several previous works \cite{Mukherjee_2024,uola20} that contextuality in such a bipartite scenario has a one-to-one correspondence with the measurement incompatibility of the measurements performed at individual wings. In this paper, we investigate whether the contextuality-breaking channels have a connection with the channels that breaks measurement incompatibility.
 
Crucially, as discussed earlier, in the bipartite scenario with two measurement inputs for each party, the witnesses for nonlocality and contextuality coincide, namely the violation of the Bell-CHSH inequality. Therefore, if a set of incompatible measurements violates the Bell-CHSH inequality, it is not evident whether such measurements are witnessing contextuality or nonlocality. However, when each party has more than two inputs, the noncontextual and local bounds of the corresponding functional differs. For example, the contextuality witness for the 3322-Bell scenario is given by the violation of the noncontextual version of Elegant-Bell inequality (NEBI), i.e., $\langle\mathfrak{B}\rangle > 4$, whereas the nonlocality witness is $\langle\mathfrak{B}\rangle > 6$. Therefore, for scenarios with more than two measurement inputs, it becomes possible to distinguish the relationship between measurement incompatibility and these two types of quantum correlations.
 
 \textit{Incompatibility breaking channel (IBC):}  If a quantum channel for the set of input observables $A_1, \dots, A_n$ $(n\geq2)$, produces the outputs $\mathcal{E}^*(A_1), \dots, \mathcal{E}^*(A_n)$ that are compatible, then it is said to be \textit{incompatibility breaking} \cite{heinosaari2015}.
If a channel $\mathcal{E}$ can make every set of $N$ observables compatible, it is referred to as $N$-IBC. For example,  a channel is said to be $2$-IBC  if it breaks the incompatibility of each pair in the set of observables. We are now in a position to describe our main results that establish a comparative relation between the aforementioned channels.
 \section{Results}

 \subsection{Contextuality Breaking Conditions}
 \label{subsec:cbreaking}
 To begin, we formally define CBC through the contextuality-breaking condition expressed in terms of a specific contextuality witness.
 
\begin{definition}
\label{def:NCbreaking}
    A channel is said to be contextuality-breaking, if after application of the channel on one of the qubits, the joint statistics satisfies NEBI, i.e., the states $\rho_{AB}^{1}=(\mathcal{E}\otimes\mathbf{I})(\rho_{AB})$, and $\rho_{AB}^{2}=(\mathbf{I}\otimes\mathcal{E})(\rho_{AB})$ satisfy the  $ \langle \mathfrak{B}_{Q}^{opt} \rangle \leq4$, with $\mathfrak{B}_{Q}^{opt}$ is the quantum optimal value of the Elegant-Bell functional.
\end{definition} 

We can now derive the contextuality-breaking condition for an arbitrary channel. In order to do this, consider the two-qubit maximally entangled state expressed in the Pauli basis,
\begin{eqnarray}
\label{statepauli}
\rho_{AB}=\frac{1}{4}(\mathbb{I}\otimes\mathbb{I}+\sigma_x\otimes\sigma_x-\sigma_y\otimes\sigma_y+\sigma_z\otimes\sigma_z)
\end{eqnarray}
The Depolarizing Channel (DPoC) applied on the Alice's qubit transforms the state in Eq.(\ref{statepauli}) as,
\begin{eqnarray}
\hspace{-0.6cm}(\mathcal{E}\otimes\mathbf{I})(\rho_{AB})=\rho_{AB}' = p\rho_{AB}+(1-p)\frac{\mathbb{I}}{2} \otimes \tr_{A}(\rho_{AB})
\end{eqnarray}

From Eq.\eqref{eq:EBI} we get the expectation value of the EBI-Bell functional to be $\langle \mathfrak{B}_{Q}^{opt} \rangle =  \max_{A_{x},B_{y}}(\tr[\rho_{AB}'\mathfrak{B}_{Q}])= 4\sqrt{3}p$.
Here, we take the same observables given in Eq. \eqref{eq:alice_obs} and Eq. \eqref{eq:bob_obs} that gives the optimal quantum violation of the $\mathfrak{B}_{Q}$ to find the threshold value of the channel parameter $p$.
Hence, from Definition \ref{def:NCbreaking}, we obtain the condition under which the channel is CBC as follows:
 \begin{eqnarray}
 \langle \mathfrak{B}_{Q}^{opt} \rangle = 4\sqrt{3}p \leq4.
 \end{eqnarray}
 which implies,
 \begin{eqnarray}
 \label{NBC}
 p\leq \frac{1}{\sqrt{3}}.
 \end{eqnarray}
Similarly, considering different types of channels, namely amplitude damping (ADC), loss (LC), and dephasing channels (DPC), we have derived the conditions on the channel parameters under which the channel behaves like a CBC. The specific Krauss representations of the channels are given in the Appendix \ref{ap:krauss}. This was obtained by applying the action of each channel on the first qubit, the second qubit, and both qubits, as summarized in the Table.~\ref{EBINBC} below. In all cases, the quantum expectation value of the EBI-Bell functional remains the same, regardless of whether the channel is applied to the first or the second qubit. On the other hand, if we apply the channel to both qubits, it decreases faster with $p$ compared to the case when it is only applied to a single subsystem, which is exactly what one intuitively expects. The EBI Bell functional $\mathfrak{B}_{Q}$ is plotted as a function of the channel parameter $p$ for different channels, applied to either a single qubit or both qubits, in Fig.~\ref{fig:ebi}.

\begin{table}[h!]
\centering
\begin{tabular}{@{}lccc@{}}
\toprule
\textbf{Channel} & \textbf{First qubit} ($\mathbf{B_Q}$) & \textbf{Second qubit} ($\mathbf{B_Q}$) & \textbf{Both qubits} ($\mathbf{B_Q}$) \\
\midrule
$DPoC$   & $4\sqrt{3}p$         & $4\sqrt{3}p$         & $4\sqrt{3}p^2$            \\
$ADC$   & $\dfrac{4(2\sqrt{p} + p)}{\sqrt{3}}$ & $\dfrac{4(2\sqrt{p} + p)}{\sqrt{3}}$ & $\dfrac{4(1 + 2p^2)}{\sqrt{3}}$                \\
$LC$   & $4\sqrt{3}p$                         & $4\sqrt{3}p$                         & $4\sqrt{3}p^{2}$                              \\
$DPC$  & $\dfrac{12 - 8p}{\sqrt{3}}$          & $\dfrac{12 - 8p}{\sqrt{3}}$          & $\dfrac{12 + 8(p - 2)p}{\sqrt{3}}$           \\
\bottomrule
\end{tabular}
\caption{Maximum quantum value of EBI-Bell operator ($\mathfrak{B}_{Q}$) for different types of noise applied on the first, second, and both qubits of the state.}
\label{EBINBC}
\end{table}

 We compare the CHSH-CBC and the EBI-CBC for different channels in Table \ref{CHSHEBINBC2}. Note that CHSH-NBC has previously been studied in Ref. \cite{pal2015non,Zhang2020}. It is intuitively straightforward to see, involving a greater number of observables requires the channel to introduce more distortion into the state to break the concerned quantum correlation. We found that for all types of channels listed in the Table \label{CHSHEBINBC2}, there exists a range of channel parameters for which the channel is capable of breaking the CHSH-nonlocality but unable to break the EBI-contextuality.

\setlength{\tabcolsep}{20pt} 
\renewcommand{\arraystretch}{2} 

\begin{table}[h!]
\centering
\begin{tabular}{@{}lcc@{}}
\toprule
\textbf{Channel} & \textbf{CHSH-CBC} & \textbf{EBI-CBC} \\
\midrule
$DPoC$   & $p \leq \dfrac{1}{\sqrt{2}}$     & $p \leq \dfrac{1}{\sqrt{3}}$     \\
$ADC$   & $p \leq \dfrac{1}{2}$            & $p \leq 0.4262$                  \\
$LC$   & $p \leq \dfrac{\sqrt{5} - 1}{2}$ & $p \leq \dfrac{1}{\sqrt{3}}$     \\
$DPC$  & $p \leq \dfrac{1}{\sqrt{2}}$     & $p \leq 0.6339$                  \\
\bottomrule
\end{tabular}
\caption{Contextuality-Breaking Conditions under Different Noise Channels Using Witnesses Associated with the CHSH and Elegant-Bell Functionals.}
\label{CHSHEBINBC2}
\end{table}
  
For the depolarizing channel, in the range of the channel parameter $\frac{1}{\sqrt{3}}\leq p \leq \frac{1}{\sqrt{2}}$, the channel is CHSH-CBC (or equivalently CHSH-NBC) but not EBI-CBC. Similar implications hold for others channels also as summarized in Table. \ref{CHSHEBINBC2}. These implications can be traced from the general results Theorem \ref{thm:EBI-CHSH} and Corollary \ref{thm:EBI-CHSH} stated and proved bellow. But before going into those results we prove the following Lemma that is required for the proof.

\begin{figure}%
		\centering
		\subfloat[\centering ]{{\includegraphics[width=9cm]{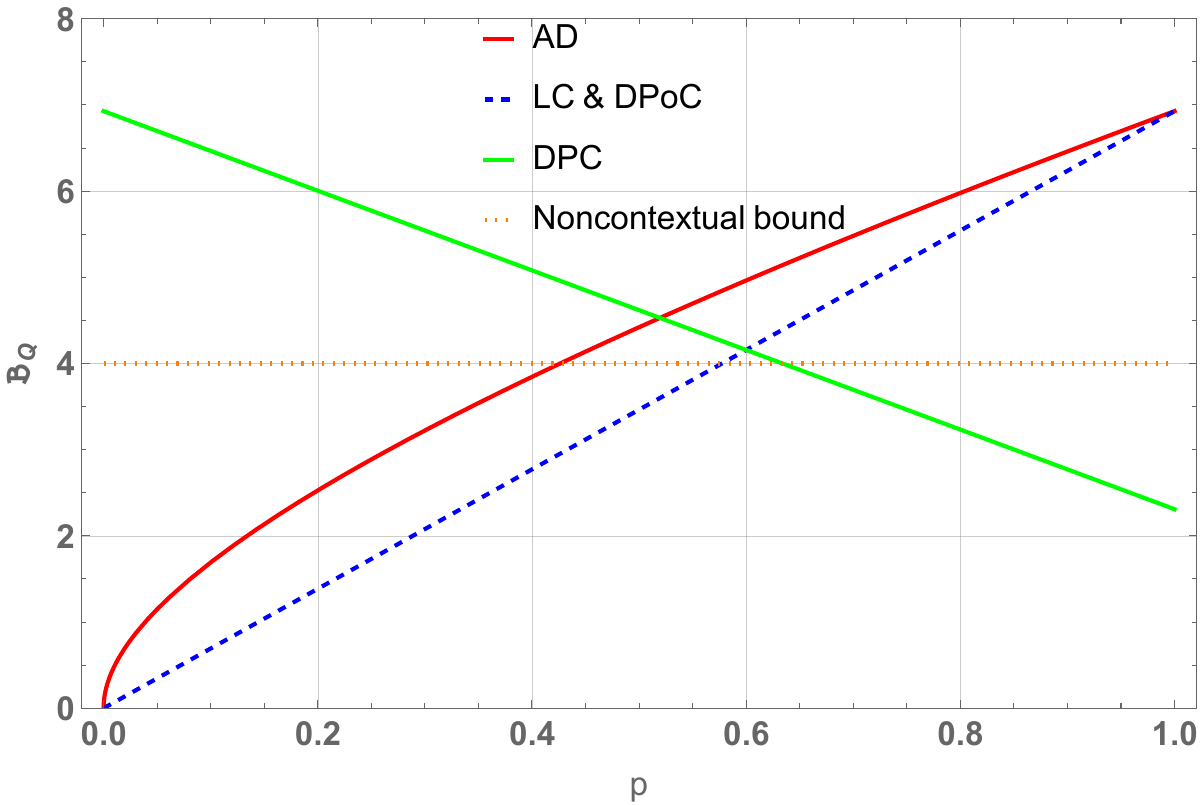} }}%
		\qquad
		\subfloat[\centering ]{{\includegraphics[width=9cm]{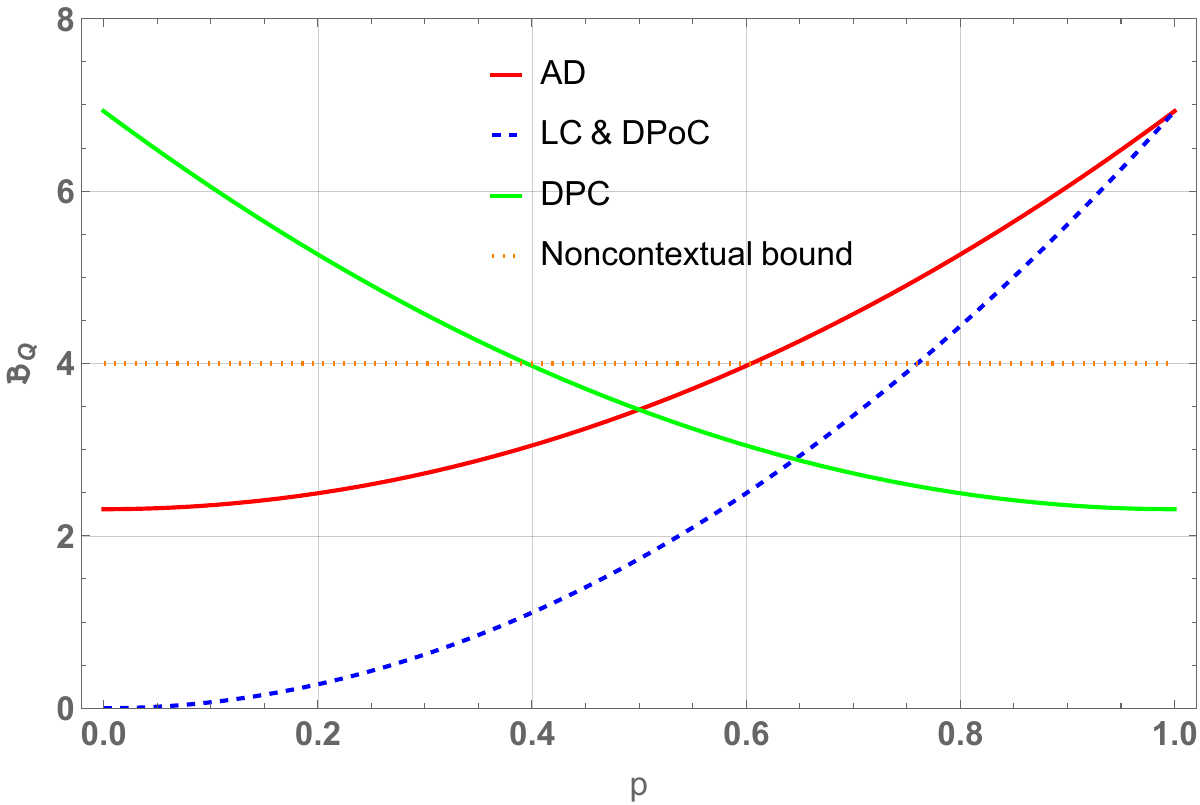} }}%
		\caption{Plot of quantum value of Bell functional $\mathfrak{B}_{Q}$ as a function of the channel parameter $p$ with the channel applied to (a) a single qubit, (b) both the qubits.}%
		\label{fig:ebi}%
	\end{figure}

\begin{Lemma}
\label{lem:povmchannel}
Let $\Lambda$ be a unital qubit channel, and let $\Lambda^{*}$ denote its Heisenberg-picture dual. Then, for every unbiased qubit POVM element
\begin{equation}
 E=\frac{1}{2}\left(\mathbb{I}+\mathbf{n}\cdot\boldsymbol{\sigma}\right), \qquad |\mathbf{n}|\leq 1,   
\end{equation}
the operator $\Lambda^{*}(E)$ is also an unbiased qubit POVM element and can be expressed as
\begin{equation}
\Lambda^{*}(E)=\frac{1}{2}\left(\mathbb{I}+\mathbf{n}'\cdot\boldsymbol{\sigma}\right),    \end{equation}
where $\mathbf{n}'\in\mathbb{R}^{3}$ satisfies $|\mathbf{n}'|\leq 1$.
\end{Lemma}

\medskip

\noindent
\textbf{Proof.} Any $2\times 2$ Hermitian operator $A$ can be written in terms of the generators of $\mathbf{SU}(2)$ groups as,
\begin{equation}
 A = a_0 \mathbb{I} + \sum_{i=1}^3 a_i \sigma_i,   
\end{equation}

where $\{\sigma_1,\sigma_2,\sigma_3\} \equiv \{\sigma_x,\sigma_y,\sigma_z\}$ are the Pauli matrices.
  
Now, since a quantum channel $\Lambda$ is trace-preserving, in the Heisenberg picture this implies $\Lambda^*(\mathbb{I}) = \mathbb{I}$, i.e., $\Lambda^*$ is unital. Furthermore, if $\Lambda$ is unital, then for any operator $T$, we have,
\begin{align}
    \operatorname{tr}[\Lambda^*(T)]
    =\operatorname{tr}[T\Lambda(\mathbb{I})]
    =\operatorname{tr}[T\mathbb{I}]
    =\operatorname{tr}[T],
\end{align}
In other word, the Heisenberg-picture dual $\Lambda^*$ of a unital channel $\Lambda$ is trace-preserving. Therefore, as each $\sigma_i$ is traceless and Hermitian, $\Lambda^*(\sigma_i)$ is also traceless and Hermitian, and thus can be written as a real linear combination of Pauli matrices:
\begin{equation}
\Lambda^*(\sigma_i) = \sum_{j=1}^3 R_{ji}\,\sigma_j,
\end{equation}
for some real $3\times 3$ matrix $R$.

Furthermore, we can write $E$ in terms of the Pauli matrices as,
\begin{equation}
E = \frac{1}{2}\mathbb{I} + \frac{1}{2}\sum_{i=1}^3 n_i \sigma_i.
\end{equation}
Then from linearity of quantum channels we have,
\begin{equation}
\begin{split}
E' &= \Lambda^*(E) 
     = \frac{1}{2}\Lambda^*(\mathbb{I}) + \frac{1}{2} \sum_{i=1}^3 n_i\,\Lambda^*(\sigma_i) \\
   &= \frac{1}{2}\mathbb{I} + \frac{1}{2} \sum_{i=1}^3 n_i \sum_{j=1}^3 R_{ji} \sigma_j \\
   &= \frac{1}{2}\mathbb{I} + \frac{1}{2} \sum_{j=1}^3 \left( \sum_{i=1}^3 R_{ji}n_i \right) \sigma_j.
\end{split}
\end{equation}
Defining the new Bloch vector $\mathbf{n}'$ via $n'_j = \sum_{i=1}^3 R_{ji}n_i$ (i.e., $\mathbf{n}' = R\,\mathbf{n}$), we have
\[
E' = \frac{1}{2}\big(\mathbb{I} + \mathbf{n}'\cdot\boldsymbol{\sigma}\big).
\]
  
It is important to note that since $\Lambda^{*}$ unital, it maps the Bloch ball into itself, implying $|\mathbf{n}'| \le 1$ whenever $|\mathbf{n}| \le 1$.  
Therefore $E'$ is again a valid effect.

\hfill$\square$

\medskip

\noindent
\textbf{Example.} The Depolarizing channel is a unital channel. Under the action of this channel a state transforms as,
\begin{equation}
\Lambda(\rho) = (1-p)\rho + p\,\frac{\mathbb{I}}{2},
\end{equation}
we have $R = (1-p)\mathbb{I}_3$, so $\mathbf{n}' = (1-p)\mathbf{n}$ and
\begin{equation}
E' = \frac{1}{2}\big(\mathbb{I} + (1-p)\,\mathbf{n}\cdot\boldsymbol{\sigma}\big),  
\end{equation}
which is of the same form.

\begin{thm}
A unital quantum channel $\mathcal{E}$ is Elegant-Bell inequality contextuality-breaking (EBI-CBC) if and only if its dual $\mathcal{E}^*$ is 3-incompatibility-breaking (3-IBC). In other words,
\begin{equation}
\mathcal{E} \text{ is EBI-NBC } \Longleftrightarrow \mathcal{E}^* \text{ is 3-IBC}.
\end{equation}
\end{thm}

\begin{proof}
    
 To prove this, let us consider a channel $\mathcal{E}$ is EBI-CBC, then from Definition \ref{def:NCbreaking}, for the post-channel action quantum state $\rho_{AB}'=(\mathcal{E}\otimes\mathbb{I})(\rho_{AB})=(\mathbb{I}\otimes\mathcal{E})(\rho_{AB})$  will satisfy the following inequality,
\begin{eqnarray}
\tr[\rho_{AB}'\mathfrak{B}]\leq4
    \end{eqnarray}
    which implies,
    \begin{align}
 \label{NLI2}
\nonumber  \tr[(\mathcal{E}\otimes\mathbf{I})(\rho_{AB})(A_{1}&\otimes(B_{1}+B_{2}-B_{3}-B_{4})\\
 \nonumber +A_{2}&\otimes(B_{1}-B_{2}-B_{3}+B_{4})\\
 +A_{3}&\otimes(-B_{1}-B_{2}+B_{3}+B_{4}) ]\leq4 
 \end{align}
 Following Eq.\eqref{eq:SHpictures}, in Schrondinger picture we can write the above equation as,
  \begin{align}
 \label{NLI2}
 \nonumber \tr[\rho_{AB}(\mathcal{E^{*}}A_{1})&\otimes(B_{1}+B_{2}-B_{3}-B_{4})\\
\nonumber  +(\mathcal{E^{*}}A_{2})&\otimes(B_{1}-B_{2}-B_{3}+B_{4})\\
 +(\mathcal{E^{*}}A_{3})&\otimes(-B_{1}-B_{2}+B_{3}+B_{4}) ]\leq4. 
 \end{align}
 
 By optimizing the LHS of Eq. \eqref{NLI2} over measurements on Bob's side, and for the unbiased qubit observables $A_{x}=\frac{\mathbb{I}+\vec{a}_{x} \cdot \sigma }{2}$ on Alice's side we get,
  \begin{eqnarray}
&&|\vec{a}'_{1}+\vec{a}'_{2}+\vec{a}'_{3}|+|\vec{a}'_{1}+\vec{a}'_{2}-\vec{a}'_{3}|+|\vec{a}'_{1}-\vec{a}'_{2}-\vec{a}'_{3}| \nonumber \\
&&  \hspace{4.2cm} +|\vec{a}'_{1}-\vec{a}'_{2}+\vec{a}'_{3}|\leq4,
 \end{eqnarray}
 
 Here, we used the transformed operator $\mathcal{E}^{*}A_{x}=\tfrac{1}{2}\left(\mathbb{I}+\vec{a}'_{x}\cdot\sigma\right)$, which has already been derived in Lemma \ref{lem:povmchannel}. This is the condition for compatibility of three unbiased qubit observables \cite{kunjwalMS,KunjwalJM}. Hence, $\mathcal{E}^{*}$ is 3-$\tibc$.

 To prove the converse consider $\mathcal{E}^*$ is 3-IBC. Then, for any triple of input observables $\{A_1, A_2, A_3\}$, the outputs $\{\mathcal{E}^*(A_1), \mathcal{E}^*(A_2), \mathcal{E}^*(A_3)\}$ are jointly measurable. It is known that for the EBI scenario, a violation requires measurement incompatibility at Alice’s side \cite{Mukherjee_2024}. Since the incompatibility is broken by $\mathcal{E}^*$, no violation of EBI can occur. Therefore, for any state $\rho_{AB}$,
\begin{eqnarray}
    \langle \mathfrak{B}_{Q}^{opt} \rangle &&= \tr[(\mathcal{E} \otimes \mathbf{I})(\rho_{AB}) \cdot \mathfrak{B}_{Q}] \nonumber \\
    &&= \tr[(\rho_{AB}) \cdot (\mathcal{E^{*}} \otimes \mathbf{I})\mathfrak{B}_{Q}] \leq 4,
\end{eqnarray}

and hence $\mathcal{E}$ is EBI-CBC.
 
 \end{proof}
 
 \begin{cor}
\label{thm:EBI-CHSH}
  If a qubit channel $\mathcal{E}$ is EBI-CBC, then it is CHSH-CBC, but the converse is not true.  
\end{cor} 
 
 \begin{proof}
     It is now clear that the necessary and sufficient conditions for a channel $\mathcal{E}$ to be EBI-CBC is that its dual $\mathcal{E}^{*}$ is 3-IBC. Furthermore, as we already know that  $3-\tibc \subseteq 2-\tibc$ \cite{heinosaari2015}, then $\mathcal{E}^{*}$ is also CHSH-CBC. However, all CHSH-CBC may not be the EBI-CBC. This can be understood from the following example. In the case of the depolarizing noise, the range of channel parameter $\frac{1}{\sqrt{3}}\leq p\leq\frac{1}{\sqrt{2}}$, the channel is 2-$\tibc$ but not 3-$\tibc$. This means that the channel $\mathcal{E}$ is CHSH-CBC but not EBI-CBC. 
 \end{proof}

\subsection{White-noise robustness of contextuality under channels} \label{sec:white_robustness}

To quantify the robustness of EBI-based contextuality under a quantum channel, we introduce the notion of white-noise robustness, defined as the maximum amount of white noise that must be added to the post-channel state for it to satisfy the NEBI for all measurements. For simplicity, we assume that the state shared between Alice and Bob is a pure two qubit maximally entangled state $\ket{\Phi^+}$. Then for a post-channel state $\rho'=(\mathcal{E}\otimes\mathbf{I})(\ket{\Phi^+}\bra{\Phi^+})$, with the channel being applied to the subsystem at Alice's wing let,
\begin{equation}
\langle\mathfrak{B}(\rho')\rangle=\max_{A_x,B_y}\tr[\rho' \mathfrak{B}] 
\end{equation}
be the maximal value of the Bell functional in EBI for \(\rho'\). With the maximum noncontextual value $\langle\mathfrak{B}_{NC}\rangle=4$, we define the white-noise robustness of contextuality as,
\begin{equation}
    R_{\mathrm{WN}}(\rho')=\max\Big(0,1-\frac{\langle\mathfrak{B}_{NC}\rangle}{\langle\mathfrak{B}(\rho')\rangle}\Big)
=\max\Big(0,1-\frac{4}{\langle\mathfrak{B}(\rho')\rangle}\Big).
\end{equation}
The value of $R_{\mathrm{WN}}(\rho')$ determines the amount of noise that must be mixed with the post-channel state in order for the resulting statistics to become noncontextual. Therefore, a larger value of $R_{\mathrm{WN}}(\rho')$ indicates stronger contextuality of the correlation.
We now discuss the white-noise robustness of contextuality, quantified by $R_{\mathrm{WN}}(p)$, for the four specific channels considered in Sec.~\ref{subsec:cbreaking}, using the corresponding analytic expressions for $\langle\mathfrak{B}(\rho')\rangle$.

\paragraph{Depolarizing channel (DPoC).}
For the depolarizing channel $\langle\mathfrak{B}_{DP}(\rho') \rangle =4\sqrt{3}\,p$, hence,
\begin{equation}
    R_{\mathrm{WN}}^{\mathrm{DPoC}}(p)=\max\!\Big(0,\,1-\frac{1}{\sqrt{3}\,p}\Big).
\end{equation}
Then the robustness is nonzero only for \(p>1/\sqrt{3}\).

\paragraph{Amplitude-damping channel (ADC).}
Using the closed form expression for the quantum value of Bell functional under ADC action on Alice's qubit from Table \ref{EBINBC}, we write $
\langle\mathfrak{B}_{ADC}(\rho') \rangle=\frac{4}{\sqrt{3}}\big(2\sqrt{p}+p\big)$. Therefore,
\begin{equation}
    R_{\mathrm{WN}}^{\mathrm{AD}}(p)=\max\!\Big(0,\,1-\frac{\sqrt{3}}{2\sqrt{p}+p}\Big).
\end{equation}

The EBI-CBC breaking condition gives $2\sqrt{p}+p=\sqrt{3}$, which implies that the robustness in nonzero from a threshold value of channel parameter \(p_{\mathrm{th}}^{\mathrm{ADC}}\approx 0.4262\) ADC.

\paragraph{Loss channel (LC).}
For LC the EBI value has the same dependence as the depolarizing case as, $\langle\mathfrak{B}_{LC}(\rho') \rangle=4\sqrt{3}\,p$, hence,
\begin{equation}
  R_{\mathrm{WN}}^{\mathrm{LC}}(p)=\max\!\Big(0,\,1-\frac{1}{\sqrt{3}\,p}\Big).  
\end{equation}

Therefore, in this case the robustness is nonzero only for $p>1/\sqrt{3}$.

\paragraph{Dephasing channel (DPC).}
For the dephasing channel we have the analytic expression from Table \ref{EBINBC} as, $\langle\mathfrak{B}_{DPC}(\rho') \rangle=\frac{12-8p}{\sqrt{3}}$, and therefore the expression for the robustness takes the form,

\begin{equation}
R_{\mathrm{WN}}^{\mathrm{DPC}}(p)=\max\!\Big(0,\,1-\frac{4\sqrt{3}}{12-8p}\Big).  
\end{equation}

Then the the white-noise robustness is nonzero only for $p<p_{\mathrm{th}}^{\mathrm{DPC}} =(3-\sqrt{3})/2\approx 0.634$.

The trend of the white-noise robustness $R_{WN}(p)$ under the action of different channels as a function of channel parameter $p$ is depicted in Fig. \ref{fig:robust}. To provide a quantitative picture, the values of the robustness for specific values of $p$ are summarized in Table~\ref{tab:RWN_compact}.

\begin{table}[h!]
\centering
\renewcommand{\arraystretch}{2}
\setlength{\tabcolsep}{8pt}
\begin{tabular}{c|ccccc}
\hline
$p$ & 0.0 & 0.40 & $1/\sqrt{3}$ & 0.70 & 1.00 \\
\hline
$R_{\mathrm{WN}}^{\mathrm{DP}}$ & 0.000 & 0.000 & 0.000 & 0.175 & 0.423 \\
$R_{\mathrm{WN}}^{\mathrm{AD}}$ & 0.000 & 0.000 & 0.174 & 0.270 & 0.423 \\
$R_{\mathrm{WN}}^{\mathrm{LC}}$  & 0.000 & 0.000 & 0.000 & 0.175 & 0.423 \\
$R_{\mathrm{WN}}^{\mathrm{DPC}}$  & 0.423 & 0.213 & 0.061 & 0.000 & 0.000 \\
\hline
\end{tabular}
\caption{The table shows white-noise robustness $R_{\mathrm{WN}}(p)$ for depolarizing, 
amplitude-damping, loss, and dephasing channels for different values of $p$.}
\label{tab:RWN_compact}
\end{table}

\begin{figure}
     \includegraphics[height=50mm, width=85mm,scale=1.5]{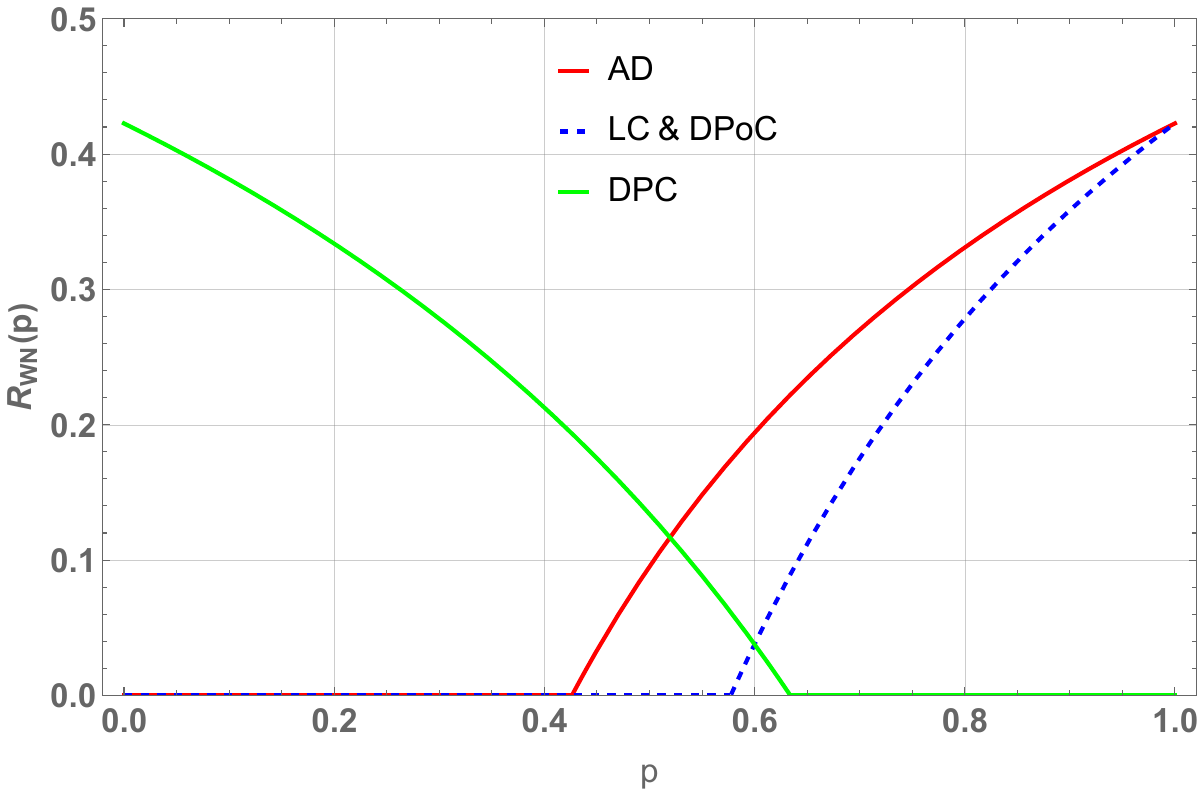}
    \caption{Plots of the white-noise robustness $R_{WN}(p)$ as a function of the channel parameter $p$ for Different Channels acting on the subsystem of Alice's wing.}
    \label{fig:robust}
\end{figure}

\subsection{Generalization to $N$-incompatibility breaking Channels}
\label{subsec:gen}
We have already established the relationship between channels that break triple-wise incompatibility and those that break contextuality in a $(3422)$-Bell scenario. While the witness of contextuality in this specific scenario is given by the EBI, its generalization \cite{Mukherjee_2024} serves as a witness for any $(n2^{n-1}22)$-Bell scenario. In such a case, violation of the following inequality acts a witness for contextual quantum correlation.

\begin{equation}\label{jgm}
    \mathcal{B}_{N}= \max_{B_{y}}\Big(\sum_{y=1}^{2^{N-1}}\lvert \sum_{x=1}^{N} (-1)^{l_{y}^{x}}\langle A_{x} \otimes B_{y} \rangle \rvert\Big) \leq 2^{N-1} 
\end{equation}
where $l_{y}^{x}$ is the $x^{th}$  bit sampled from $y^{th}$ string from the bit-string $l\in\{0,1\}^n$ having first entry $0$. The quantum optimal value of the operator $\mathcal{B}_{N}$ is given by,

\begin{align}
\label{opt}
(\mathcal{B}_{N})_{Q}^{max}= p 2^{N-1}\sqrt{N}.
\end{align}

There exists operators of certain form \cite{Mukherjee_2024} for a shared maximally entangled state and $N$ number of mutually anti-commuting operators at Alice's side, such that this optimal quantum value is achieved. The connection between the contextuality-breaking channels and $N$-incompatibility-breaking channels can be best summarized by the following result.

\begin{thm}
    A depolarizing channel $\mathcal{E}$ with channel parameter $p$ acting on Alice's part of a shared entangled state in a Bell scenario, where Alice performs $N$ unbiased POVMs and Bob performs $2^{N-1}$ random measurements, breaks contextuality if its dual channel $\mathcal{E}^{*}$ breaks $N$-wise incompatibility.
\end{thm}

\begin{proof}
    In order to prove the theorem we first describe the necessary condition for $N$-wise  compatibility of $N$ unbiased quantum measurements of Alice. Let us take dichotomic observables $A_{x}$ where $A_{x} = \{A_{x}^{a_{x}}|A_{x}^{a_{x}} > 0; \sum_{a_{x}}A_{x}^{a_{x}}= \mathbb{I}\}$. As already discussed, $\{A_{x}\}_{x}$ are jointly measurable if there exists a global POVM $\{G_{1,2,....,N}^{a_{1},a_{2},...,a_{N}}\}_{N}$ such that,

\begin{equation}
	A_{k}^{a_{k}}=\sum_{a_{1},..a_{k-1}, a_{k+1},..a_{N}}G_{1,2,....,N}^{a_{1},a_{2},...,a_{N}}(a_1,\ldots,a_k,\ldots,a_N)
\end{equation}
Following Lemma \ref{lem:povmchannel}, we can write a unbiased POVM $E_{x}^{a_{x}}$ after the action of a depolarizing channel as,

\begin{equation}
\label{eq:evpovm}
\mathcal{E}  A_{x}^{a_{x}}= \bar{A}_{x}^{a_{x}} = \frac{\mathbb{I} + p\, a_{x} \Omega_{x}}{2},
\end{equation}  
where $\Omega_{x}$ are the generators of the Clifford algebra satisfying the following,

\begin{equation}
\label{anticom}
    \forall x, x', \ \ \Omega_{x}\Omega_{x'} +  \Omega_{x'}\Omega_{x}= 2\delta_{x,x'} \mathbb {I},
\end{equation}

It has been established in ref. \cite{KunjwalJM} that, for $N$ binary POVMs of the form of Eq.\eqref{eq:evpovm} a necessary and sufficient condition for $N$-wise compatibility is  
\begin{equation}
\label{nec1}
    p \leq 1/\sqrt{N} .
\end{equation}

On the other hand, from Eq. \eqref{jgm} we can see that the necessary condition for a contextuality-breaking channel, i.e., $(\mathcal{B}_{N})_{Q}^{max} \leq 4$ can be written as,

\begin{eqnarray}
\label{nec2}
    p && \leq \frac{2^{N-1}}{   \max_{B_{y}}\Big(\sum_{y=1}^{2^{N-1}}\lvert \sum_{x=1}^{N} (-1)^{l_{y}^{x}}\langle A_{x} \otimes B_{y} \rangle \rvert\Big)\Big) } \nonumber \\
    &&=\frac{2^{N-1}}{2^{N-1}\sqrt{N}} = \frac{1}{\sqrt{N}}
\end{eqnarray}

Therefore, a depolarizing channel $\mathcal{E}$ applied on a subsystem of a bipartite state is contextuality-breaking if the dual channel $\mathcal{E}^{*}$ on the unbiased measurement acting on the same subsystem is $N$-wise incompatibility breaking.
\end{proof}

Similarly, the relation between CBC and IBC can be investigated for other relevant channels that introduce noise beyond the depolarizing noise. 

\section{Conclusive remarks and future directions}
\label{sec:discussion}

In this work, we have introduced and analyzed the concept of \emph{Elegant-Bell inequality contextuality-breaking channels} (EBI-CBC) for qubits and explored their connection to \emph{incompatibility-breaking channels} (IBC). By considering an asymmetric bipartite Bell scenario with three dichotomic measurements on Alice’s side and four on Bob’s, we established a rigorous link between EBI contextuality-breaking and triple-wise measurement incompatibility breaking channels. One of our main results shows that any channel that is EBI-CBC is necessarily \emph{CHSH-nonlocality-breaking} (CHSH-NBC), but the converse does not hold. This reveals a hierarchy of channel-induced correlation destruction, where breaking contextuality in more complex measurement scenarios demands a stronger disturbance than that required for nonlocality.

For specific channels including depolarizing, amplitude damping, loss, and dephasing, we determined the exact parameter regimes separating the behavior of EBI-CBC from CHSH-NBC. Notably, in the case of loss channels, we observed a one-to-one correspondence between EBI-CBC and triple-wise IBC, reinforcing the operational equivalence between contextuality and measurement incompatibility in certain noise models. These findings refine our understanding of the interplay between these two fundamental resources and their degradation under quantum noise.

From a broader perspective, our results underline that the complexity of the measurement scenario acts as a ``resource filter": channels capable of destroying simple forms of nonclassicality may fail to eliminate more demanding ones. This insight has implications for designing robust quantum information protocols where resilience against different noise types depends not just on the amount of entanglement but also on the measurement structure in particular the incompatible set of measurements.

We note that several natural extensions of the present work are worth studying in future. Firstly, generalization to higher-order incompatibility: Our Theorem~2 establishes a direct connection between EBI-CBC and triple-wise IBC by considering unbiased POVMS under the effect of depolarizing channel. Extending this to general $n$-wise incompatibility breaking channels, and exploring whether a similar equivalence holds for other families of contextuality inequalities beyond EBI, remains an open problem. Secondly, the idea of \emph{channel activation} where multiple copies or auxiliary systems restore a broken nonclassicality has been studied for CHSH nonlocality. Extending this to EBI contextuality could reveal whether contextuality can be activated more easily or with different resources than nonlocality. Lastly, since entanglement, nonlocality, and incompatibility are related but distinct resources, mapping out the overlap and separation between entanglement-breaking, nonlocality-breaking, and contextuality-breaking channels may yield a unified resource-theoretic classification. This calls for further investigation.\\

 \section{Acknowledgements}
RP acknowledges financial support from the Science and Engineering Research Board (SERB), Government of India, under the CRG/2021/008795 project grant. SM and SK acknowledge the Institute postdoctoral fellowship from IIT Dharwad. SK acknowledges Alok Kumar Pan  for his fruitful comments on the initial manuscript of this work and the hospitality during the IIT Hyderabad visit in March 2025. 
 
 \bibliographystyle{apsrev4-1}
\bibliography{Paper}

\appendix
\section{The particular forms of Krauss operators of different channels used}
\label{ap:krauss}

\textbf{Depolarizing Channel:} A depolarizing channel with probability $p$ replaces the state with the maximally mixed state $\frac{\mathbb{I}}{2}$ and with probability $(1-p)$ leaves it unchanged.  
The Kraus operators of the qubit depolarizing channel are:
\begin{eqnarray}
K_{0} &=& \sqrt{\frac{1+3p}{4}} \, \mathbb{I}, \quad
K_{i} = \sqrt{\frac{1-p}{4}} \, \sigma_{i}, \quad i=1,2,3
\end{eqnarray}
with $\sigma_{i}$'s are Pauli matrices in the $x,y$, and $z$ directions. Its easy to verify that they are trace-preserving as $\sum_{i=0}^3 K_i^\dagger K_i = \mathbb{I}$.

\vspace{0.5em}

\textbf{Amplitude Damping Channel:}  
The amplitude damping channel models energy dissipation, driving the qubit toward the ground state $|0\rangle$, with $|0\rangle$ represents the ground state of spin in $z$ direction.
Its Kraus operators are:
\begin{align}
K_0 &= \frac{1+\sqrt{p}}{2}\,I+\frac{1-\sqrt{p}}{2}\,\sigma_z, \nonumber \\
K_1 &= \frac{\sqrt{1-p}}{2}\left(\sigma_x+i\sigma_y\right).
\end{align}
Here, $p$ is the probability of decay from $|1\rangle$ to $|0\rangle$.

\vspace{0.5em}

\textbf{Loss Channel:}  
A loss channel transmits the state perfectly with probability $p$ and replaces it with a fixed state $|0\rangle$ with probability $(1-p)$.  
In two dimensional subspace the Krauss operators are given by:
\begin{align}
K_0 &= \sqrt{p}\,I, \nonumber \\
K_1 &= \frac{\sqrt{1-p}}{2\sqrt{2}}\left(\sigma_x+i\sigma_y\right), \nonumber \\
K_2 &= \frac{\sqrt{1-p}}{2\sqrt{2}}\left(I-\sigma_z\right).
\end{align}
For photonic systems with an explicit vacuum mode, additional Kraus operators would be needed to include the orthogonal loss subspace.

\vspace{0.5em}

\textbf{Dephasing Channel:}  
The dephasing channel leaves populations unchanged but suppresses off-diagonal elements by a factor $(1-p)$. The map is:
\begin{equation}
\mathcal{E}(\rho) = (1-p)\rho + p \, \mathrm{diag}(\rho_{00}, \rho_{11})
\end{equation}
Its Kraus operators can be written as:
\begin{align}
K_0 &= \sqrt{1-p}\,I, \nonumber \\
K_1 &= \frac{\sqrt{p}}{2}\left(I+\sigma_z\right), \nonumber \\
K_2 &= \frac{\sqrt{p}}{2}\left(I-\sigma_z\right).
\end{align}

\end{document}